\documentclass[a4paper,12pt]{article}
\usepackage{amsmath}
\usepackage{amsthm}
\usepackage{amsfonts}
\usepackage{amssymb}
\usepackage{mathrsfs}
\usepackage{graphicx}
\usepackage{color}
\usepackage[colorlinks=true, allcolors=blue]{hyperref}

\newcommand{\lR}{\mathrm{I\hspace{-0.7mm}R}}

\newtheorem{theorem}{Theorem}
\newtheorem{lemma}{Lemma}

\newtheorem{corollary}{Corollary}

\usepackage[english]{babel}

\numberwithin{equation}{section}

\hypersetup{ pdftitle={Draft},
pdfauthor={Fiki T. Akbar, Bobby E. Gunara}, 
pdfkeywords={Local Esistence, Dyonic Black Hole}, 
bookmarksnumbered, pdfstartview={FitH}, urlcolor=blue, }

\begin{document}
\pagestyle{plain}




\title{\LARGE\textbf{  Static  Black Holes of  Higher Dimensional Einstein-Skyrme System with General Couplings}}

\author{\normalsize{Bobby Eka Gunara}$^{\sharp}$\footnote{Corresponding author}, Emir Syahreza Fadhilla$^{\sharp, \ddagger}$, 
and 
Ardian Nata Atmaja$^{\ddagger}$ \\ \\
$^{\sharp}$ \textit{\small Theoretical High Energy Physics Research Division,}\\
\textit{\small Faculty of Mathematics and Natural Sciences,}\\
\textit{\small Institut Teknologi Bandung}\\
\textit{\small Jl. Ganesha no. 10 Bandung, Indonesia, 40132}
\\ {\small and} \\
$^{\ddagger}$\textit{\small Research Center for Quantum Physics, National Research and Innovation Agency (BRIN),}\\
\textit{\small Kompleks PUSPIPTEK Serpong, Tangerang 15310, Indonesia}\\
\\
\small email: bobby@itb.ac.id, emirsyahreza@students.itb.ac.id, ardi002@brin.go.id}

\date{\today}

\maketitle




\begin{abstract}

We construct  hairy static black holes of higher dimensional general coupling Einstein-Skyrme theories with the scalar potential turned on and the cosmological constant is non-positive in which the  scalar multiplets satisfy $O(d+1)$ model constraint where $d$ is the spatial dimension of the spacetime and  $d \ge 3$.  Some physical properties of solutions near the boundaries, namely, near the (event) horizon and in the asymptotic region are discussed. Then, we prove that these black hole solutions exist globally which may have finite energy for non-positive cosmological constant. Finally,  we use perturbative method to perform a linear dynamical stability analysis and then, show the existence of stable and unstable solutions in the model.

\end{abstract}




\section{Introduction}
\label{sec:intro}

Black holes in Einstein-Skyrme model are the evidence of  counterexamples to the no-hair conjecture for black holes  which states that a  black hole can  only be characterized by its mass and its electric and magnetic charges. These objects up to  now have been studied mostly  in four dimensions where the Skyrme field is an $SU(2)$ valued. In particular, several models admitting  asymptotically flat spacetimes exist in the literature, see for example,  \cite{Volkov:2016ehx}.  We have only few models with  asymptotically anti-de Sitter spacetimes \cite{Shiiki:2005aq, Shiiki:2005xn, Perapechka:2016cof}, while there is only a model with  de Sitter background \cite{Brihaye:2005an}.

There are  some higher dimensional models which can be mentioned as follows.  A five dimensional Einstein-Skyrme model has been considered which can be thought of as  an $O(5)$ sigma model coupled to gravity \cite{Brihaye:2017wqa}. In the model, we have  some universality properties that can be adopted to any dimension higher than five. For example, a topological charge is globally  defined on the spacetime. The second model is a seven dimensional Skyrme  branes \cite{BlancoPillado:2008cp} which provides  a brane world scenario. The authors set that  the skyrmion field defines  on a warped spherically symmetric three dimensional submanifold of  a seven-dimensional spacetime which is conformal to $\lR^{3+1} \times {\mathcal S}^3$ where $\lR^{3+1}$ and $ {\mathcal S}^3$ are four-Minkowski spacetime and three-conformally flat geometries, respectively. Another example of  black holes in higher dimensional  Einstein-Skyrme theories is studied in \cite{Gunara:2018lma}. In this latter model, the authors study the Einstein-Skyrme theory  with the cosmological constant $\Lambda \le 0$ on a $d+1$-dimensional nontrivial static spacetime $\mathcal{M}^{d+1}$ which is conformal to  $ \mathcal{M}^{3+1} \times \mathcal{N}^{d-3}$ where $\mathcal{M}^{3+1}$ and $\mathcal{N}^{d-3}$ are  the four dimensional spacetime  and the compact  $(d-3)$-dimensional submanifold, respectively. The theory that the authors considered is up to the fourth term and the Skyrme field is still an $SU(2)$ valued. Then, they construct a family of static black holes and prove the existence of such solutions with finite energy if  $\Lambda = 0$.

The purpose of this paper is to provide the analysis  of static black hole solutions of higher dimensional  Einstein-Skyrme models with general couplings including  the scalar potential $V(\phi)$ and the cosmological constant $\Lambda \le 0$. In particular, we set  the spacetime $\mathcal{M}^{d+1}$ with  $d \ge 3$ to be static and  conformal to $ \mathcal{M}^{1+1} \times S^{d-1}$ where $\mathcal{M}^{1+1}$ and $S^{d-1}$ are  the two dimensional spacetime  and the compact  $(d-1)$-dimensional sphere, respectively, with the metric functions $\delta(r)$ and $m(r)$. We also have  an $O(d+1)$ valued Skyrme field which is locally defined on the submanifold  $\mathcal{S}^{d} \subseteq \mathcal{M}^{d+1} $, where $\mathcal{S}^{d}$ is conformal to $\lR^+ \times S^{d-1} $  with $\lR^+ $   being positive real number. This Skyrme field can further be simplified into a form that can be expressed  in terms of a profile function  $\xi \equiv \xi(r)$ where $r$ is  the radial coordinate \cite{Date:1986be}. 

We write down  some consequences of the above as follows. First, the covariant topological current such as baryon number lives locally  on the submanifold  $\mathcal{S}^{d}$.  Second,  to obtain  physical black hole, we have to specify the behavior of the functions $\delta(r), ~ m(r)$,  and $\xi(r)$  on the boundaries,  namely, at the (event)  horizon and the outer boundary, together with  their local-global existence and their  linear stability. Finally, we want to mention that since our model consists of general couplings in diverse dimension ($d \ge 3$) which is a complicated structure, it is not necessary to use the notion of branches related to  the value of  $\xi(r)$ on the horizon.  Instead, we just need  the value of  $\xi(r)$ on the horizon to be regular (or finite) which will be useful to establish the local-global existence of solutions. 
  

 In order to have a well-defined model, we first have to analyze the functions $\delta(r)$, $m(r)$, and $\xi(r)$ near the boundaries.   Near the (event) horizon, all the function can be linearly expanded such that these functions become fixed  on the horizon. At this region, the spacetime $\mathcal{M}^{d+1}$ breaks into  ${\mathcal T}^{1+1} \times S^{d-1} $  where the 2-surface ${\mathcal T}^{1+1}$ could be either a flat Minkowski surface $\lR^{1+1}$ or  an anti-de Sitter surface $AdS_2$ \cite{Kunduri:2007vf}. Analysis on the Ricci scalar implies that the value of  $\xi(r)$ on the horizon has to satisfy an inequality in order to have a consistent solution. In the asymptotic limit, we set the decay rate of all the functions $\delta(r), ~ m(r)$, and $\xi(r)$ to be of the form $O(r^{-n})$ with $n \ge 1$ such that the lowest order of these functions are constants implying that  the black hole spacetime converges Einstein geometry.

Next, we construct local-global existence and uniqueness of this skyrmionic black hole solution. By employing  Picard's iteration and the contraction mapping properties, we firstly show local existence and uniqueness of the solutions. Then, using the uniqueness property we argue that the local solution can be extended to the maximal solution by gluing some of the local solutions. Since we have particularly the decay properties of the functions $\delta(r), ~ m(r)$,   and $\xi(r)$  of the form $O(r^{-n})$ with $n \ge 1$ in the asymptotic region, it can be established a family of  global solutions whose energy is finite by taking the values of both the field $\xi(r)$ and the scalar potential to be vanished in this region for arbitrary    $\Lambda \le 0$. 

 Finally, we discuss the so called linear dynamical stability of solutions using perturbation method to obtain a linear equation called Sturm-Liouville equation in which the problem is reduced to  an eigenvalue problem. It is important to notice that the stability here is in the context of Sturm-Liouville theory. Using the behavior  of the functions $\delta(r), ~ m(r)$,   and $\xi(r)$  in the asymptotic region and applying so called the fixed point theorem \cite{Krasnosel_1964}, it can be shown the existence of both stable and unstable solutions  for any $\Lambda \le 0$.


We organize this  paper  as follows. In Section \ref{sec:Skyrmemod} we discuss shortly the Skyrme model in diverse dimension.  The static solutions in the theory are considered in Section \ref{sec:EinsteinSkyrmemod}. We perform the analysis  of the solutions near the boundaries, namely, the (event) horizon and the asymptotic region in Section \ref{sec:boundaryprop}. In Section \ref{sec:ExisSolFinitE} we prove that there exists a family of unique global  solutions with finite energy in the model. In Section  \ref{sec:linearstableanal} we consider the dynamical linear stability of solutions and show the existence of stable and unstable solutions.  

\textbf{Conventions and Notations}.
In this paper we use alphabets, \(a,b,c,\dots\), as scalar multiplets indices and mid alphabets, \(i,j,k,\dots\) as spacetime indices. Both indices sets take value from \(0\) to \(d\) with \(d+1\)is the spacetime dimensionality. A one-form component is represented by \(\omega_i\) with lower index and its dual is represented by \(\omega^i\) with upper index. The spacetime coordinates are denoted by \(x^i\) and its derivative operator is \(\frac{\partial}{\partial x^i}\), or sometimes written as \(\partial_i\). The multiple scalar  index is written as an upper index and the square bracket \([\dots]\) on indices represents antisymmetrization.

\section{Skyrme Model in Diverse Dimension}
\label{sec:Skyrmemod}


In this section we shortly discuss the Skyrme model in $d \ge 3$ spatial dimensions. The starting point is to consider the standard  Skyrme model in \(3+1\) dimension whose action is of the form
\begin{equation}\label{eq:SkyrmeL4d}
    \mathcal{S}_4= \int d^4x ~ \sqrt{-g}\left( \gamma_1g^{ij}~Tr(L_i L_j)+\gamma_2g^{ij}g^{kl}~Tr([L_i, L_k][L_j,L_l])\right) ~ ,
\end{equation}
with \(L_i=U^{\dagger}\partial_i U\) where \(U=\phi^0I_{2}+\bold{i}\vec{\phi}.\vec{\sigma}\) is an \(SU(2)\) valued chiral field originally  proposed by T. Skyrme \cite{Skyrme:1961vq, Skyrme:1962vh}.  The vector \(\vec{\sigma}=(\sigma_1,\sigma_2,\sigma_3)\) are the Pauli matrices, while  $(\phi^0,\vec{\phi})$ are real scalar field satisfying \(O(4)\) model condition, $ \phi^a\phi^a=1$, see also for example,  \cite{Brihaye:2017wqa, Manton:2004tk}. 

In fact, one can work another way  to construct such functional by introducing a so called strain tensor  \(D=JJ^T\) where \(J\) is the Jacobian matrix of the map \cite{manton1987}.  The energy functional of static Skyrme model can be constructed from the invariants of the strain tensor $D$ which are the combinations of its eigenvalues. Then, the Lagrangian from \eqref{eq:SkyrmeL4d} can be expressed to be the sum of these two terms
\begin{eqnarray}
    \phi^a_i\phi^a_jg^{ij} ~ , \\
    \phi^a_{[i}\phi^b_{j]}\phi^a_{[k}\phi^b_{l]}g^{ij}g^{kl} ~ ,
\end{eqnarray}
where \(\phi^a_i \equiv \frac{\partial\phi^a}{\partial x^i}\) and \(g^{ij}\) are the components of the inverse of metric tensor. This construction of Skyrme model can be employed  to  generalize \eqref{eq:SkyrmeL4d} either for including the higher order terms \cite{Gudnason_2017} or for extending the Skyrme model in five dimensions \cite{Brihaye:2017wqa}. 


Now, let us consider  the eigenvalues of \(D_{ij}\) from a \(d+1\) dimensional Skyrme model, that is, \(\lambda_1^2,\lambda_2^2,\dots,\lambda_d^2\). The most general Lagrangian can be written as
\begin{equation} \label{eq:SkyrmeLdgen}
    {\mathcal L}  = - \gamma_0 V-\sum_{n=1}^d \gamma_n\mathcal{L}_{2n} ~ ,
\end{equation}
where $\gamma_p  \ge 0$, $p = 0,...,d$,  are the coupling constants, \(V\equiv V(\phi)\) is a scalar potential, and \(\mathcal{L}_{2n}\)  have the form
\begin{eqnarray}
    \mathcal{L}_{2}&\propto& \lambda_1^2+\dots+\lambda_d^2 ~ , \nonumber\\ 
    \mathcal{L}_{4}&\propto& \lambda_1^2\lambda_2^2+\lambda_1^2\lambda_3^2+\dots+\lambda_{d-1}^2\lambda_d^2  ~ , \nonumber\\
    \vdots \nonumber\\
    \mathcal{L}_{2d}&\propto& \lambda_1^2\lambda_2^2\dots\lambda_d^2 ~ .
\end{eqnarray}
In four dimensional case (\(d=3\)), we have  the standard Skyrme model  with additional BPS-Skyrme term (\(\mathcal{L}_{6}\propto \lambda_1^2\lambda_2^2\lambda_3^2\)) proposed in \cite{Adam:2010zz}.

Using the above prescription and identifying
\begin{eqnarray}\label{eq:SkyrmeLdgenterm}
    \mathcal{L}_{2n}  = \frac{1}{\left(n!\right)^2} \phi^{a_1}_{[i_1}\dots\phi^{a_n}_{i_n]}\phi^{a_1}_{[j_1}\dots\phi^{a_n}_{j_n]}g^{i_1j_1}\dots g^{i_nj_n} ~ ,
\end{eqnarray}
the Lagrangian of the model has the form
\begin{equation}\label{eq:SkyrmeLdgenform}
    {\mathcal L} =\gamma_0 V+\sum_{n=1}^d \frac{\gamma_n}{\left(n!\right)^2}\phi^{a_1}_{[i_1}\dots\phi^{a_n}_{i_n]}\phi^{a_1}_{[j_1}\dots\phi^{a_n}_{j_n]}g^{i_1j_1}\dots g^{i_nj_n} ~ ,
\end{equation}
which is just an \(O(d+1)\) model  proposed in \cite{Brihaye:2017wqa}. 



\section{Static Solutions in Einstein-Skyrme Model}
\label{sec:EinsteinSkyrmemod}

We can now couple the above Skyrme to  gravity via Einstein-Hilbert Lagrangian
\begin{equation}\label{eq:EinsteinSkyrmeLdgenform}
    {\mathcal L}_g = \frac{1}{2} (R - 2 \Lambda ) - \gamma_0 V-\sum_{n=1}^d \frac{\gamma_n}{\left(n!\right)^2}\phi^{a_1}_{[i_1}\dots\phi^{a_n}_{i_n]}\phi^{a_1}_{[j_1}\dots\phi^{a_n}_{j_n]}g^{i_1j_1}\dots g^{i_nj_n} ~ ,
\end{equation}
where $R$ and $\Lambda$ are the Ricci scalar and the cosmological constant, respectively, defined on $(d+1)$-dimensional spacetime ${\mathcal M}^{d+1}$. Throughout the paper, we take $\Lambda \le 0$. In order to be compatible with the discussion in the preceding section, we take the ansatz metric to be static spherical symmetric
\begin{equation}\label{eq:metricans}
    ds^2=-e^{2\delta}f ~ dt^2+\frac{dr^2}{f} + r^2  R_0^2 ~ d\Omega^2_{d-1} ~ ,
\end{equation}
where \(\delta\equiv\delta(r)\), \(f\equiv f(r)\), and \(d\Omega^2_{d-1}\) is the metric of $(d-1)$-dimensional sphere $S^{d-1}$ with radius $R_0 > 0$. 
For the Skyrme field, we use generalized hedgehog ansatz \cite{Date:1986be} which has the following form
\begin{equation}\label{ansatz}
    \phi=(\cos\xi,\Vec{n}\sin\xi)
\end{equation}
 with \(\Vec{n}.\Vec{n}=1\), \(\vec{n}\) is a vector in \(\mathbb{R}^d\). In fact, the ansatz \eqref{ansatz} is a natural choice for a spherically symmetric system. We can motivate this as follows: such ansatz is a coordinate system on target space with the following metric \(ds_{target}^2=d\xi^2+\sin^2\xi d\Omega_{d-1}^2\) where \(\Vec{n}.\Vec{n}=1\) directly represents \(S^{d-1}\) manifold with metric \(d\Omega_{d-1}^2\). For a spherically symmetric spacetime with metric of the form 
\(ds^2=-A(r)dt^2+\frac{1}{f(r)}dr^2+r^2d\Omega_{d-1}^2\) we can immediately observe that both manifold have the same submanifold because both contain \(S^{d-1}\) manifold on which the spherical symmetry applies. Furthermore, if we chose the field \(\xi\) to be \(\xi(r)\), the map \(\phi\) simply become a coordinate transformation \(r\rightarrow\xi(r)\). From this point of view we expect that the angular dependencies of \(\phi\) are decoupled in the full field equations, hence leaves only a single spherical symmetric field equation for \(\xi\).

Here, we show the details on how we can calculate the effective lagrangian by utilizing eigenvalues of strain tensor \(D\) with \(D^i_j=g^{ik}\phi^a_k\phi^a_j\) as its components. The component of tensor \(D\) in mixed form for the ansatz mentioned above is given by
\begin{equation}
    D^i_j=g^{ik}\partial_k\xi\partial_j\xi+\sin^2\xi g^{ik}\partial_k\Vec{n}\cdot\partial_j\Vec{n}
\end{equation}
To make the calculation simpler, we can choose a spatial coordinate on which the strain tensor become diagonal. Firstly, we choose the same parameterization for both \(S^{d-1}\) submanifold of spacetime and target space of \(\phi\). Hence, suppose that \(S^{d-1}\) submanifold of spacetime has the following metric
\begin{equation}
    d\Omega_{d-1}^2=d\theta_1^2+\sum_{p=2}^{d-1}\prod_{q=1}^{p-1}\sin^2\theta_q~d\theta_p^2
\end{equation}
then the vector \(n\) must be a function of \(\theta_1,\dots,\theta_{d-1}\) and \(\partial_k\Vec{n}\cdot\partial_j\Vec{n}\) is metric tensor component of \(S^{d-1}\) except \(\partial_1\Vec{n}\cdot\partial_1\Vec{n}\). Now, let us define the metric tensor on \(S^{d-1}\) as \(h_{\theta_p\theta_q}\) and choose \(\xi\) to be \(\xi\equiv\xi(r)\). We can immediately observe that the spherical components of spacetime metric \(g_{\theta_p\theta_q}\) is related to the metric of sphere manifold by \(g_{\theta_p\theta_q}=r^2 h_{\theta_p\theta_q}\). Therefore, our strain tensor components become
\begin{eqnarray}
D^1_1&=&g^{11}\xi'^2=f\xi'^2,\\
D^{\theta_p}_{\theta_q}&=&\sin^2\xi\frac{h^{\theta_p\theta_s}}{r^2}h_{\theta_s\theta_q}=\frac{\sin^2\xi}{r^2}\delta^{\theta_p}_{\theta_q}
\end{eqnarray}
The resulting tensor is diagonal, thus we can directly identify all the eigenvalues, namely
\begin{eqnarray}
\lambda_1^2&=&f\xi'^2,\\
\lambda_2&=&\dots=\lambda_d=\frac{\sin^2\xi}{r^2}
\end{eqnarray}
The next step is to formulate each invariant from combinations of eigenvalues. We know that for \(\mathcal{L}_{2n}\) there are \(C^d_n\) terms. Because only one eigenvalue is different, which is \(\lambda_1^2=f\xi'^2\), then there are \(C^{d-1}_{n-1}\) terms which have \(f\xi'^2\) as their factor. From this we can further conclude that there are \(C^d_n-C^{d-1}_{n-1}\) terms with no factor of \(f\xi'^2\), hence the invariant \(\mathcal{L}_{2n}\) is given by
\begin{equation}
    \mathcal{L}_{2n}=\frac{\sin^{2(n-1)}\xi}{r^{2(n-1)}}\left( C^{d-1}_{n-1}f(\xi')^2+\left(C^{d}_{n} - C^{d-1}_{n-1}\right)\frac{\sin^2\xi}{r^2}\right)~.
\end{equation}
This formulation can be used in a more general \(D^{\theta_i}_{\theta_j}\) which represents a non-trivial map to the target space \(\Vec{n}\cdot\Vec{n}=1\). 
The field equations for Skyrmion can be found by substituting ansatz \eqref{ansatz} to the dynamical equations of the Skyrme model, namely
            \begin{equation}
                \left(\delta^{cb}-\phi^c\phi^b\right)\left[\nabla_k\left(\sum_{n=1}^d \frac{2\gamma_n}{\left(n!\right)^2}g^{i_1j_1}\dots g^{i_n j_n}\phi^{a_1}_{[i_1}\dots\phi^{a_n}_{i_n]}\frac{\partial}{\partial \phi^b_k}\phi^{a_1}_{[j_1}\dots\phi^{a_n}_{j_n]}\right)-\frac{\partial}{\partial \phi^b} \gamma_0 V\right]=0~.
            \end{equation}
The dynamical field equations above can be simplified by exploiting Levi-Civita symbol and strain tensor, leading to the following expression of field equations
\begin{eqnarray}\label{eq: fieldEQ}
\left(\delta^{cb}-\phi^c\phi^b\right)\left[\sum_{n=1}^d \frac{2n\gamma_n}{n!(d-n)!}\nabla_{j}\left(\phi^{b,i}H_i^j(n,d)\right)-\frac{\partial}{\partial \phi^b} \gamma_0 V\right]=0~,
\end{eqnarray}
where we have defined the tensor \begin{equation}H_i^j(n,d)\equiv\varepsilon_{i~i_2\dots i_n k_{n+1}\dots k_d}\varepsilon^{j~j_2\dots j_n k_{n+1}\dots k_d}\nabla_{j_1}\prod_{m=2}^n D_{j_m}^{i_m}\end{equation} for computational purpose. \(H_i^j(n,d)\) has special features, namely, it is a diagonal tensor if the strain tensor is diagonal and only depends on radial coordinate if we substitute hedgehog ansatz in \eqref{ansatz} to \(\phi\). Upon substitution of ansatz \eqref{ansatz}, the resulting field equations \eqref{eq: fieldEQ} is effectively reduced into a single field equation (as expected, because \(\Vec{n}\) decouples) for \(\xi\) where all of the equations only differ by some angular factors, namely
\begin{equation}\label{eq:eomSkyrme}
    \xi''+\left(\frac{d-1}{r}+\delta'+\frac{f'}{f}+\frac{1}{u}\partial_r u  \right)\xi'+  \frac{(\xi')^2}{2u}\frac{\partial u}{\partial \xi}  =\frac{1}{2uf} \partial_\xi v ~ ,
\end{equation}
with $\partial_r u \equiv \frac{\partial u}{\partial r}$ and $\partial_\xi v \equiv \frac{\partial v}{\partial\xi}$.
where we have defined
\begin{eqnarray}\label{eq:uv}
    u &\equiv& \sum_{n=1}^d\gamma_n C^{d-1}_{n-1} \frac{\sin^{2(n-1)}\xi}{(rR_0)^{2(n-1)}}  ~ , \nonumber\\
    v &\equiv&  \gamma_0V+\sum_{n=1}^d \gamma_n C^{d}_{n} \left(1-\frac{n}{d}\right) \frac{ \sin^{2n}\xi}{(rR_0)^{2n}}  ~ ,
\end{eqnarray}
and  we assume $V = V(\xi)$ for the rest of this paper.
On the metric \eqref{eq:metricans}, the effective static energy functional for the Skyrme model has the form
\begin{equation}\label{eq:EeffSkyrme}
    E=\int d^{d}x \sqrt{-g^{(d+1)}}  \left( u f (\xi')^2+v \right) ~ ,
\end{equation}
 Varying the energy \eqref{eq:EeffSkyrme} with respect to $\xi$, leads to the same field equation of $\xi$  given by \eqref{eq:eomSkyrme}.

Next, the dynamics of the metric functions $\delta(r)$ and $f(r)$ are governed by the Einstein field equation coming from the variation of the action related to Lagrangian \eqref{eq:EinsteinSkyrmeLdgenform}. The components of the Einstein field equation can be simplified into the following equations
\begin{equation}
  (d-1)\frac{f}{2}\left( \frac{f'}{rf}+\frac{d-2}{r^2}\right)-\frac{(d-1)(d-2)}{2 r^2 R^2_0} + \Lambda  = -   u f (\xi')^2 - v    ~ ,   
  \label{eq:Einsteineq}    
\end{equation}
\begin{equation}
  (d-1)\frac{f}{2}\left( \frac{2\delta'f+f'}{rf}+\frac{d-2}{r^2}\right) -\frac{(d-1)(d-2)}{2 r^2 R^2_0} + \Lambda  = u f (\xi')^2 - v    ~ ,  
\label{eq:Einsteineq1}    
\end{equation}
\begin{eqnarray}
  && \frac{f}{2}\left( 2\delta''+2(\delta')^2 + \frac{f''+3\delta'f'}{f}\right) + \frac{d-2}{r}\left(f'+\delta'f\right)-\frac{(d-2)(d-3)}{2 r^2 R^2_0}\left(1-fR_0^2\right) + \Lambda \nonumber\\
 &&  \quad  = ~  \frac{1}{d (d-1)} \sum_{n=1}^d\gamma_n C^{d}_{n} ~ n \left(2n-d-1\right) \frac{\sin^{2(n-1)}\xi}{(rR_0)^{2(n-1)}}    f(\xi')^2   \nonumber\\
 && \quad \quad - \frac{1}{d (d-1)} \sum_{n=1}^d\gamma_n C^{d}_{n} \left( 2n^2-3nd+n-d+d^2\right) \frac{\sin^{2n}\xi}{(rR_0)^{2n}}  - \gamma_0 V  ~ .
\label{eq:Einsteineq2}    
\end{eqnarray} 
Additionally, we write down  the Ricci scalar of the metric \eqref{eq:metricans}
\begin{eqnarray} \label{eq:Ricciscalar}
R &=&\frac{1}{r^2 R^2_0} (d-2)(d-1) -\frac{2 }{r} (d-1)(f'+\delta' f) \nonumber\\
&&-\left[ 2f(\delta''+\delta'^2) + 3\delta' f'+f'' \right]  - (d-2)(d-1) \frac{f}{r^2} ~ .
\end{eqnarray}
which will be useful in the next section discussion.

Now, we take a particular for of the metric function $f(r)$ as 
\begin{equation}\label{eq:metform}
    f= \frac{1}{R_0^2}-\frac{2m}{(d-2)r^{d-2}}-\frac{2\Lambda}{d(d-1)}r^2 ~ ,
\end{equation}
where $m \equiv m(r)$.  Using \eqref{eq:metform}, \eqref{eq:Einsteineq} and \eqref{eq:Einsteineq1} become simply   
\begin{eqnarray}
    m' &=& \frac{(d-2)}{(d-1)}r^{d-1} \left(u f (\xi')^2+v \right) ~ , \nonumber\\
    \delta' &=& \frac{2}{(d-1)}u(\xi')^2r  ~ .\label{eq:constraints}
\end{eqnarray}
We will discuss the properties of  $m(r)$, $\delta(r)$, and $\xi(r)$ near the boundaries, namely, near the horizon ($r \to r_h$) and around the asymptotic region ($r \to +\infty$) in the next section. 

\section{Near Boundary Properties}
\label{sec:boundaryprop}

In this section we discuss the properties of solutions of the Einstein field equation \eqref{eq:Einsteineq2}, \eqref{eq:constraints}, and the equation of motions \eqref{eq:eomSkyrme} near the boundaries. The first part of discussion is to consider the properties of $m(r)$, $\delta(r)$, and $\xi(r)$ near the horizon and then, we continue it in the asymptotic region as $r \to +\infty$.

Let us first consider the behavior of $m(r)$, $\delta(r)$, and $\xi(r)$  in the horizon limit. Suppose there exists a horizon at the radius $r = r_h$ such that $f(r_h) = 0$ implying
\begin{equation} \label{eq:massBH}
 m(r_h) = \frac{(d-2)r_h^{d-2}}{2R_0^2}-\frac{(d-2)}{d(d-1)}\Lambda r_h^d ~ .
\end{equation}
 Near the region, we can expand the metric functions  as
\begin{eqnarray}
\xi &=& \xi_h + \xi_1(r-r_h)+O((r-r_h)^2) ~ , \nonumber\\
\delta &=& \delta_h + \delta_1(r-r_h)+O((r-r_h)^2) ~ , \nonumber\\
m &=& m(r_h) + m_1(r-r_h)+O((r-r_h)^2) ~ , \label{eq:horizonexpand}
\end{eqnarray}
where $\xi_h$ and $\delta_h$ are positive constants, while $\xi_1$, $\delta_1$, and $m_1$ are real constants. Inserting \eqref{eq:horizonexpand} into  \eqref{eq:constraints} and \eqref{eq:eomSkyrme}, it then yields
\begin{eqnarray}
\delta_1 &=& \frac{2}{(d-1)}\left(\sum_{n=1}^d\gamma_n C^{d-1}_{n-1}\frac{\sin^{2(n-1)}\xi_h}{(r_hR_0)^{2(n-1)}} \right)\xi_1^2r_h ~ ,\nonumber\\
m_1 &=& \frac{(d-2)}{(d-1)}r_h^{d-1}\left( \gamma_0V(\xi_h)  +\sum_{n=1}^d \gamma_n C^{d}_{n} \left(1-\frac{n}{d}\right) \frac{\sin^{2n}\xi_h}{(r_hR_0)^{2n}}  \right)  ~ , \nonumber\\
 \xi_1 &=& \frac{\gamma_0 ~ \partial_\xi V(\xi_h)+\frac{\sin(2\xi_h)}{(r_hR_0)^2}\sum_{n=1}^d \gamma_n C^{d-1}_{n-1} (d-n)\frac{\sin^{2(n-1)}\xi_h}{(r_hR_0)^{2(n-1)}} }{2\left(\frac{d-2}{R_0^2r_h}-\frac{2\Lambda r_h}{(d-1)}-\frac{2m_1}{(d-2)r_h^{d-2}}\right)\left(\sum_{m=1}^d\gamma_n C^{d-1}_{m-1} \frac{\sin^{2(m-1)}\xi_h}{(r_hR_0)^{2(m-1)}} \right)} ~ , \label{eq:horizonexpandres}
\end{eqnarray}
which show that we have a $(d+4)$-dimensional parameter space spanned by $\xi_h$, $\delta_h$,  $M$, and $\gamma_p$ with $p = 0,...,d$.  From the last equation in \eqref{eq:horizonexpandres}, it is straightforward to see that among $\gamma_p$,  we should have at least one of them to be non-zero.

The second expression in \eqref{eq:horizonexpandres} and \eqref{eq:Einsteineq2} can be used to find the form of Ricci scalar on the horizon which takes the form
\begin{equation}\label{eq:ricciscalarhorizon}
    R|_{r_h}  = \frac{(d-1)(d-2)}{r^2_hR_0^2} + \Lambda_e  + 2 \sum_{n=1}^d \gamma_n C^{d}_{n} \left(\frac{2n^2 -3nd -n+d +d^2}{d(d-1)}\right) \frac{\sin^{2n}\xi_h}{(R_0r_h)^{2n}} ~ ,
\end{equation}
with
\begin{equation} 
    \Lambda_e  =  \frac{2(d+1)}{(d-1)} \left(\Lambda + \gamma_0 V(\xi_h)  \right) - \frac{(d-1)(d-2)}{r^2_hR_0^2}   ~ .
\end{equation}
Around this region, the spacetime topology changes to ${\mathcal T}^2 \times S^{d-1}$  where ${\mathcal T}^2$ is a  2-surface ${\mathcal T} ^2$ which could be   either  ${\mathcal T} ^2 \simeq \lR^2$ or  ${\mathcal T} ^2 \simeq AdS_2$ \cite{Kunduri:2007vf}. Thus, the second term and the third term in the right hand side of \eqref{eq:ricciscalarhorizon}  should be
\begin{equation}\label{eq:horizongeomcon}
    \Lambda_e  + 2 \sum_{n=1}^d \gamma_n C^{d}_{n}  \left(\frac{2n^2 -3nd -n+d +d^2}{d(d-1)}\right) \frac{\sin^{2n}\xi_h}{(R_0r_h)^{2n}}  \le 0 ~ .
\end{equation}

Next, we consider the behavior of $m(r)$, $\delta(r)$, and $\xi(r)$  in the asymptotic region. In order to have a finite and regular solution,  the metric functions $\delta(r)$ and $m(r)$ should be decreasing functions whose form are
\begin{eqnarray}
\delta(r) &=&    \frac{\tilde{\delta}_1}{r^{n_1}} + O\left(r^{-(n_1 +1)} \right) \ , \nonumber\\
m(r) &=&  M + \frac{\tilde{m}_1}{r^{n_2}} +O\left(r^{-(n_2 +1)} \right)  \ ,   \label{eq:Expandasymp}
\end{eqnarray}
where $M > 0$ is the ADM mass which will be discussed below, whereas $\tilde{\delta}_1 , m_1 \in \lR$,  $n_1 \ge 1$, and $n_2 \ge 1$. Moreover,  the skyrmionic scalar $\xi(r)$ would have the form  
\begin{equation}
\xi(r) = \xi_\infty   +  \frac{\tilde{\xi}_1}{r^{n_3}}+ O\left(r^{-(n_3 + 1)} \right) ~ , \label{eq:fasympcon}
\end{equation} 
where $n_3 \ge 1$ and $\xi_\infty, \tilde{\xi}_1 \in \lR$ showing that  $\xi(r)$  will  be frozen as $r \to +\infty$. Inserting \eqref{eq:Expandasymp} and \eqref{eq:fasympcon} to the constraints \eqref{eq:constraints}, we obtain
\begin{equation}
 n_1 \ge d+1 ~ , \quad n_2 \ge 1 ~ , \quad n_3 \ge \frac{1}{2} (d+1) ~ , \label{eq:nsympconres}
\end{equation} 
with either $\gamma_0 = 0$ or $\gamma_0 > 0$ and
\begin{equation}
V( \xi_\infty ) = 0 ~ . \label{eq:Vasympconres}
\end{equation} 
We also have
\begin{eqnarray}
  \tilde{\delta}_1 &=&  - \frac{n_3^2}{n_1}  \gamma_1  \tilde{\xi}_1^2 ~ ,\nonumber\\
 \tilde{m}_1 &=&    \frac{n_3^2 ~ \Lambda}{d (d-1)n_2}  \gamma_1  \tilde{\xi}_1^2  ~ .
 \label{eq:constraintsasymp}
\end{eqnarray}
The equation of motions \eqref{eq:eomSkyrme} gives us the condition
\begin{equation}
\gamma_0 ~ \partial_\xi V  (\xi_\infty) = \frac{4  \gamma_2 \Lambda}{d(d-1) R_0^2}   \sin \xi_\infty \cos \xi_\infty ~ ,   \label{eq:eomSkyrmeasymcon}
\end{equation}
where we have used the ansatze \eqref{eq:Expandasymp} and \eqref{eq:fasympcon}. 

Now, we write down the trace of the Einstein field equation
\begin{eqnarray}
R &=& \frac{2 (d+1)}{(d-1)}  \left(  \Lambda + \gamma_0 V\right) +  \frac{2}{d (d-1)} \sum_{n=1}^d\gamma_n C^{d}_{n} \left( 2n^2 -3nd -n+d +d^2\right)\frac{\sin^{2n}\xi}{(rR_0)^{2n}} \nonumber\\
 && - \frac{2}{d (d-1)} \sum_{n=1}^d\gamma_n C^{d}_{n} n \left(2n-d -1\right)  \frac{\sin^{2(n-1)}\xi}{(rR_0)^{2(n-1)}}    f(\xi')^2    ~ . \nonumber\\
 \label{eq:traceEinsteineq}
\end{eqnarray}
In this region the geometry converges to Einstein such that \eqref{eq:traceEinsteineq} should be simplified to
\begin{equation}
R = \frac{2 (d+1) \Lambda}{(d-1)}  + O\left(r^{-n} \right)   ~ ,   \label{eq:Ricciscalarasym}
\end{equation}
 with cosmological constant $2  \Lambda / (d-1)$ (or Ricci-flat with $\Lambda = 0$) such that if we subtitute  \eqref{eq:Expandasymp} and \eqref{eq:fasympcon} to \eqref{eq:traceEinsteineq}, we find either the case without the scalar potential, that is, $\gamma_0 = 0$ or \eqref{eq:Vasympconres} with $\gamma_0 > 0$.  As we will see in the next section, the finiteness of the energy functional \eqref{eq:EeffSkyrme} constrains the value of $\xi_\infty $ and the bound of $n_3$.

Let us now discuss how to obtain the mass of the black hole. This can be done by considering the Komar integral \cite{Kastor:2008cqg, Kastor:2009cqg}  
\begin{equation}
{\mathcal{K}} = \int_{\partial\Sigma} dS_{ij} \left( \nabla^i \zeta^j + \omega^{ij} \right) \  , \label{komarint}
\end{equation} 
where $\partial\Sigma$ is the boundary of a spatial hypersurface of the spacetime ${\mathcal M}^{d+1}$. We have a Killing vector  $\zeta^i$  of ${\mathcal M}^{d+1}$ and an antisymmetric tensor  $\omega^{ij}$  satisfying \cite{Gunara:2010iu}
\begin{equation}
\nabla_i  \omega^{ij} = R^j_{~ i} ~ \zeta^i \  . \label{divergenceomega}
\end{equation}
In the case at hand,  the vector $\zeta^i$ is  time-like with $\zeta^i = (1, 0,...,0)$ such that the non-zero solution of \eqref{divergenceomega} has to be
\begin{equation}
\omega^{10} = - \frac{1}{2}(f' + 2 \delta' f) + \frac{M}{ \text{Vol}(S^{d-2})} e^{-\delta} r^{-(d-1)}  R_0^{-(d-1)}  \  , \label{solusiomega}
\end{equation}
while
\begin{equation}
 \nabla^1 \zeta^0 =  \frac{1}{2}(f' + 2 \delta' f)  \ , 
\end{equation}
where $\text{Vol}(S^{d-2})$ is the volume of the unit sphere $S^{d-2}$. The infinitesimal volume $dS_{10}$ has the form
\begin{equation}
dS_{10} = \frac{1}{2} ~ e^{\delta} r^{d-1}  R_0^{d-1} ~ d\text{Vol}(S^{d-2})   \ . 
\end{equation}
Thus, we obtain the mass of our black hole 
\begin{equation}
 M_{\text{BH}} =\lim_{r \to +\infty} M = M \ . \label{BHmass}
\end{equation}

Finally, the topological charge called baryon number $B$ for such a system is 
\begin{eqnarray}\label{eq:topocharge}
    B&=&\frac{1}{\text{Vol}(S^d)}\int \sin^{d-1}\xi\left(\prod_{k=1}^{d-2}\sin^{d-1-k}\theta_k\right)~d\xi d\theta_1\dots d\theta_{d-1}\nonumber\\
    &=&\frac{\text{Vol}(S^{d-1})}{\text{Vol}(S^d)}\int \sin^{d-1}\xi~d\xi=\mathcal{B}(\xi_h)-\mathcal{B}(\xi_\infty) ~ ,
    \end{eqnarray}
where \(\mathcal{B}(\xi)\) is defined by
    \begin{eqnarray}
   \mathcal{B}(\xi) \equiv\frac{\Gamma\left(\frac{d+1}{2}\right)}{\sqrt{\pi}\Gamma\left(\frac{d}{2}\right)}\begin{cases}C^{d-1}_{\frac{d-1}{2}}\frac{\xi}{2^{d-1}}+\frac{1}{2^{d-2}}\sum\limits_{k=0}^{\frac{d-3}{2}}C^{d-1}_{k}(-1)^{k+1}\frac{\sin((d-2k-1)\xi)}{d-2k-1} ~ ,&~\text{for}~d=3,5,7,\dots\\
    \frac{1}{2^{d-2}}\sum\limits_{k=0}^{\frac{d-2}{2}}C^{d-1}_{k}(-1)^{k}\frac{\cos((d-2k-1)\xi)}{(d-2k-1)} ~ ,&~\text{for}~d=2,4,6,\dots\end{cases}\nonumber\\
\end{eqnarray}
 We have used the definition of topological charge given in \cite{Brihaye:2017wqa}.  A vacuum solution can be obtained  if we take $\xi_h = \xi_\infty$.


Let us consider an example as follows. Suppose we have $\xi_h = \pi$ and $\xi_\infty = 0$, then $\Lambda_e  \le 0$,   \eqref{eq:horizonexpandres} becomes
\begin{eqnarray}
\delta_1 &=&\frac{2\gamma_1 }{(d-1)}  \xi_1^2 r_h ~ ,\nonumber\\
m_1 &=& \frac{(d-2)}{(d-1)}r_h^{d-1} \gamma_0 V(\pi)  ~ , \nonumber\\
  \xi_1 &=&  \frac{\gamma_0 ~ \partial_\xi V(\pi) }{2 \gamma_1 \left(\frac{d-2}{R_0^2r_h}-\frac{2 r_h}{(d-1)}  \left(\Lambda -\gamma_0 V(\pi) \right)  \right)  } ~ , \label{eq:horizonexpandsol}
\end{eqnarray}
with $\gamma_1 > 0$, and the topological charge \eqref{eq:topocharge} take a unit value i.e. \(B=1\). 
It is easy to see that we have black hole solutions if $\gamma_0 > 0$ with either  $ V(\pi) \ne 0$ or $\partial_\xi V(\pi) \ne 0$. In other words, the presence of the scalar potential  term $\gamma_0 V$ plays an important role in determining the existence of black hole solutions in this theory. For example, in a model with generalized pion-mass scalar potential
\begin{equation}
 V(\xi) =  (1 - \cos\xi)^n  ~ , \quad n \ge 1 ~ ,
   \label{eq:pionpotensial}
\end{equation}
we may have a black hole solution since $V(\pi) = 2^n$ and  $\partial_\xi V(\pi) = 0$. Moreover, in the asymptotic region this model admits \eqref{eq:Vasympconres} and $\partial_\xi V  (0) = 0$ which shows that these  correspond to  finite energy solutions as we shall see in the next section.

In a case of $\gamma_0 = 0$ (or  $ V(\phi) \equiv 0$), we might have a smooth regular solution which is not a black hole.  This latter situation has been observed  in four dimensional Einstein-Skyrme theory \cite{Shiiki:2005aq}.
%

\section{Local-Global Existence of Finite Energy Solutions}
\label{sec:ExisSolFinitE}

In this section we establish local-global existence and uniqueness of black hole solutions of the theory. By employing the Picard's iteration and the contraction mapping properties, we prove the local existence and the uniqueness. Then, using maximal solution technique we construct the global existence. We use  the expansions \eqref{eq:Expandasymp}   and \eqref{eq:fasympcon} to show the finiteness of  solutions. Finally, the finite energy solutions  are discussed. 

\subsection{Local Existence and Smoothness}
\label{subsec:LocExisSmooth}

Let us define a set of dynamical variables ${\bf{w}} \equiv (  \xi ,  p_\xi)$ where $p_\xi \equiv \xi'$. We write down the constraints \eqref{eq:constraints}  in the integral form
\begin{eqnarray}
    m &=& m_h + \frac{(d-2)}{(d-1)}  r^{-\frac{2}{d-1}}  \int  r^{d-1 + \frac{2}{d-1}} \left( \left( \frac{1}{R_0^2}  - \frac{2\Lambda}{d(d-1)}r^2\right)  u  p_\xi^2+v \right) dr ~ , \nonumber\\
    \delta &=& \delta_h + \frac{2}{(d-1)} \int  r u p_\xi^2   ~ dr ~ , \label{eq:constraintsinteg}
\end{eqnarray}
where $u \equiv u(\xi, r)$ and $v \equiv v(\xi, r)$ given in \eqref{eq:uv}. Suppose $I \equiv [r, r+ \varepsilon]$ where $r \ne r_h$, $\varepsilon > 0$,  and $U\subset \lR^2$ be an open set. 
\begin{lemma}
	\label{lemmalocalLipshitzconstraint}
Suppose the scalar potential $V$ in \eqref{eq:EinsteinSkyrmeLdgenform} is at least a $C^2$-real function. Then, the metric functions $\delta({\bf{w}},r )$ and $m({\bf{w}},r )$ in   \eqref{eq:constraintsinteg}  are locally Lipschitz with respect to $\bf{w}$.	
\end{lemma}
\begin{proof}
 First, we have
 \begin{eqnarray}
\left|  m \right|_U  &\le&  \frac{(d-2)}{(d-1)}  r^{-\frac{2}{d-1}}  \int_r^{r+ \epsilon}  r^{d-1 + \frac{2}{d-1}} \left|  \left( \frac{1}{R_0^2}  - \frac{2\Lambda}{d(d-1)}r^2\right)  u  p_\xi^2+v \right| dr \nonumber\\
&\le&  \frac{(d-2)}{(d-1)}   C  r^{d-1 }  \left|  \left( \frac{1}{R_0^2}  - \frac{2\Lambda}{d(d-1)}r^2\right) u p_\xi^2+v \right|  ~ , \label{boundedm}
\end{eqnarray}
for $C > \epsilon$ which is bounded since the function  $\xi(r)$ is at least a $C^2$-real function.  Using similar argument, we can show that  $\delta$ is also bounded.  Then, for ${\bf{w}}, \tilde{\bf{w}} \in U$ 
\begin{eqnarray}
\left|  \delta ({\bf w}, r) -  \delta( \tilde{{\bf w}}, r) \right|_U  &\le&   \frac{2r}{(d-1)} C \left| u(\xi,r) p_\xi^2 - u( \tilde{\xi }, r) \tilde{p}_\xi^2 \right| ~ ,  \nonumber\\
\left| m ({\bf w}, r) -  m( \tilde{{\bf w}}, r) \right|_U  &\le&  \frac{(d-2)}{(d-1)} C r^{d-1} \left| \left( \frac{1}{R_0^2}  - \frac{2\Lambda}{d(d-1)}r^2\right)  \left( u(\xi,r)    p_\xi^2 - u( \tilde{\xi }, r)  \tilde{p}_\xi^2\right)\right|  \nonumber\\
&&   + \frac{(d-2)}{(d-1)} C r^{d-1} \left|  v(\xi,r) -  v( \tilde{\xi }, r)  \right|    ~ .
\end{eqnarray}
 Using the fact that for any smooth function ${\mathcal F}(f)$, we have locally 
\begin{equation}
{\mathcal F}( f) - {\mathcal F}(\tilde{f} ) \leq \sup_{s\in[0,1]}\left[ {\mathcal F}'( f + s(\tilde{f} - f)) \right]  (f -  \tilde{f}) ~ , \label{anyFlocal}
\end{equation}
on $U$, we get that $\delta$ and $m$ indeed satisfy the local Lipschitz condition
\begin{eqnarray}
\left|  \delta ({\bf w}, r) -  \delta ( \tilde{{\bf w}}, r) \right|_U  &\le& C_{ \delta}(|\bf{w}|, |\tilde{\bf{w}}|) | \bf{w} - \tilde{\bf{w}}| ~ ,  \nonumber\\
\left| m ({\bf w}, r) -  m( \tilde{{\bf w}}, r) \right|_U  &\le& C_m (|\bf{w}|, |\tilde{\bf{w}}|) | \bf{w} - \tilde{\bf{w}}|  ~ ,  \label{localLipshitzconstraints}
\end{eqnarray}
on an open set $U\subset \lR^2$ where $C_{ \delta}(|\bf{w}|, |\tilde{\bf{w}}|)$ and $C_m(|\bf{w}|, |\tilde{\bf{w}}|)$  are bounded positive-valued functions . 
\end{proof}
Next, we  rewrite  \eqref{eq:eomSkyrme} into
\begin{equation}
\frac{d \bf{w}}{dr} =  \mathcal{J}({\bf w}, r) ~ ,  \label{ReducedEinsteineqSkyrmionEq1}
\end{equation}
where  
\begin{equation}
  \mathcal{J}({\bf{w}}, r)   \equiv  \left( \begin{array}{c}
p_\xi \\
J_\xi
\end{array} \right) ~ , \label{fungsiJ}
\end{equation}
where
\begin{eqnarray}
J_\xi \equiv  - \left(\frac{d-1}{r}+\delta'+\frac{f'}{f}+\frac{1}{u}\frac{\partial u}{\partial r}\right) p_\xi - \left(\frac{1}{2u}\frac{\partial u}{\partial \xi}\right) p_\xi^2 + \frac{1}{2uf}\frac{\partial v}{\partial\xi}   ~ ,
   \end{eqnarray}
with the constraint \eqref{eq:constraintsinteg}.  We can now state the result of the local existence and the uniqueness of \eqref{ReducedEinsteineqSkyrmionEq1} as follows.
\begin{lemma}
	\label{localLipshitz}
The operator $\mathcal{J}$ defined in \eqref{fungsiJ} is locally Lipschitz with respect to $\bf{w}$.	
\end{lemma}
\begin{proof}
From \eqref{fungsiJ}, we obtain the following estimate
\begin{eqnarray}
\left|  J_\xi \right|_U  \le     \left| \frac{d-1}{r}+\delta'+\frac{f'}{f}+\frac{1}{u}\frac{\partial u}{\partial r}\right| |p_\xi | + \left| \frac{1}{2u}\frac{\partial u}{\partial \xi}\right| |p_\xi |^2 + \left| \frac{1}{2uf}\frac{\partial v}{\partial\xi} \right|  ~ .
\end{eqnarray}
Since   $\xi(r)$ belongs  at least to a class of $C^2$-real functions, then its values is bounded on any closed interval  $I$. Thus, $\left|  \mathcal{J}( {\bf w}, r) \right|_U$ is bounded on $U$.

Moreover, for ${\bf{w}}, \tilde{\bf{w}} \in U$, we also have
\begin{eqnarray}
\left|  \mathcal{J}( {\bf w}, r) -  \mathcal{J}( \tilde{ {\bf w}}, r) \right|_U & \le &  \frac{d-1}{r}  |p_\xi - \tilde{p}_\xi |  +  \left| \delta' ( {\bf w}, r) p_\xi -  \delta' ( \tilde{ {\bf w}}, r) \tilde{p}_\xi \right| \nonumber\\
&& + \left|  \frac{f'}{f} ( {\bf w}, r) p_\xi -   \frac{f'}{f} ( \tilde{ {\bf w}}, r) \tilde{p}_\xi \right|  + \left| \frac{1}{u}\frac{\partial u}{\partial r}( {\bf w}, r) p_\xi - \frac{1}{u}\frac{\partial u}{\partial r}( \tilde{ {\bf w}}, r)  \tilde{p}_\xi \right|  \nonumber\\
&& + \left| \frac{1}{2u}\frac{\partial u}{\partial \xi} ( {\bf w}, r) p_\xi^2 - \frac{1}{2u}\frac{\partial u}{\partial \xi}( \tilde{ {\bf w}}, r)  \tilde{p}_\xi^2 \right|  + \left| \frac{1}{2uf}\frac{\partial v}{\partial\xi}( {\bf w}, r) -  \frac{1}{2uf}\frac{\partial v}{\partial\xi}( \tilde{ {\bf w}}, r)  \right|        ~ .  \nonumber\\
\end{eqnarray}
Employing some computations using the result in Lemma \ref{lemmalocalLipshitzconstraint} and the local property \eqref{anyFlocal} on $U$, it can be shown $ \mathcal{J}$ is locally Lipshitz with respect to $\bf{u}$ satisfying
\begin{equation}
\left|  \mathcal{J}({\bf w}, r) -  \mathcal{J}( \tilde{{\bf w}}, r) \right|_U  \le C_{ \mathcal{J}}(|\bf{w}|, |\tilde{\bf{w}}|) | \bf{w} - \tilde{\bf{w}}| ~ .  \label{localLipshitzcon}
\end{equation}
\end{proof}

It is useful for the next analysis to  write down  \eqref{ReducedEinsteineqSkyrmionEq1} in the integral form 
\begin{equation}
{\bf{w} }(r) = {\bf{w} }(r_h) + \int_{r_h}^{r}\:\mathcal{J}\left( {\bf{w} }(s), s \right)\:ds ~  . \label{IntegralEquation}
\end{equation}
By introducing a Banach space
\begin{equation}
{\mathfrak X} \equiv \{ {\bf{w} } \in C(I,\lR^2) : \: {\bf{w} }(r_h) = {\bf{w} }_{0}, \: \sup_{r\in I}| {\bf{w} }(r)|\leq L_0 \}  ~ ,
\end{equation}
endowed with the norm
\begin{equation}
|{\bf{w} }|_{\mathfrak X} = \sup_{r\in I}\:|\mathbf{w}(r)| ~ ,
\end{equation}
where $L_0$ is a positive constant, we define an operator $\mathcal{K}$ 
\begin{equation}
\mathcal{K}(\mathbf{w}(r)) = \mathbf{w}_{0} + \int_{r_h}^{r}  ds ~  \mathcal{J}\left(s,\mathbf{w}(s)\right)\:. \label{OpKdefinition}
\end{equation}
Then, Lemma \ref{localLipshitz} implies the uniqueness result proving that the  differential equation (\ref{ReducedEinsteineqSkyrmionEq1}) has a unique local solution. 
\begin{corollary}{\textnormal{\cite{Akbar_2015}}}
\label{unigueness}
The operator $\mathcal{K}$  defined in  (\ref{OpKdefinition})  is a mapping from ${\mathfrak X} $ to itself and $\mathcal{K}$ is a contraction mapping on $I = [r,r + \varepsilon ]$ with $r \ne r_h$, $\varepsilon > 0$ ,and
\begin{equation}
\varepsilon   \leq \min\left(\frac{1}{C_{L_0}},\frac{1}{C_{L_0} L_0 + \|\mathcal{J}(r)\|}\right) ~ .
\end{equation}
Then, the operator $\mathcal{K}$ is a contraction mapping on ${\mathfrak X}$.
\end{corollary}

Finally,  a maximal solution can be constructed as follows. Let ${\bf w}(r)$ be a solution defined on the interval $( r_h, r_m)$ with $r_m > 0$. Then, we repeat the above arguments of the local existence  with the initial condition ${\bf w}(r-r_0)$ for some $r_h < r_0 < r$ and use the uniqueness condition to glue the solutions to get the maximal solution. It is straightforward that  we have a global solution by taking $r_m \to +\infty$.

\subsection{Global Existence}
\label{subsec:GlobExis}

The second part of this section we show  that  a regular global solution of \eqref{ReducedEinsteineqSkyrmionEq1} on  $I_{\infty}\equiv [r_h, +\infty)$ does exist satisfying the expansions \eqref{eq:Expandasymp} and \eqref{eq:fasympcon}. As expected, these  establish the finiteness of  \eqref{IntegralEquation} on $I_{\infty}$. 

First of all, we introduce  two intervals, namely, $I_L \equiv [r_h + \varepsilon, L]$ with $\varepsilon > 0$ and $I_A \equiv (L, +\infty)$ for finite and large $L \gg r_h$. On $I_A$, all the functions $\delta(r)$, $m(r)$, and $f(r)$ can be expanded as in  \eqref{eq:Expandasymp},   and \eqref{eq:fasympcon}. Equation  \eqref{IntegralEquation}  can be written down as
\begin{equation}
{\bf{w} }(L) = {\bf{w} }(r_h) + \int_{r_h + \varepsilon}^{L}\:\mathcal{J}\left( {\bf{w} }(s), s \right) ~ ds ~ +  \int_{L}^{+\infty}\:\mathcal{J}\left( {\bf{w} }(s), s \right) ~ ds . \label{IntegralEquation1}
\end{equation}
In order to suppress the third term in the right hand side in  \eqref{IntegralEquation1}, we should set: 1. $\gamma_1 > 0$ which means that the kinetic term must be non-zero, 2. $\partial_\xi V(\xi_\infty)$ must be well-defined. The latter condition follows from the fact that the scalar potential  $V$ should at least a $C^2$-real function as stated in Lemma \ref{localLipshitz} and it has a value given in \eqref{eq:eomSkyrmeasymcon}. Then, we have a finite and globally well-defined solution of  \eqref{IntegralEquation} since the function  $\xi(r)$ is at least a $C^2$-function.

Next, we want to prove the finite energy solutions of   \eqref{IntegralEquation} by analyzing the estimate of   the energy functional \eqref{eq:EeffSkyrme}. Then, using the expansions \eqref{eq:Expandasymp}  and \eqref{eq:fasympcon}, and taking first $\Lambda < 0$, we obtain an inequality
\begin{eqnarray} 
E & \le &  A(S^{d-1})  R_0^{d-1} \sup_{r \in I_L} \left| \int_{r_h}^L e^\delta r^{d-1} \left( u f(\xi')^2+v \right)  dr\right| \nonumber \\
 && + ~  \frac{A(S^{d-1})}{d(d-1)}  2 \gamma_1 |\Lambda | n_3^2  \tilde{\xi}^2_1 R_0^{d-1}  \left| \int_{L}^{+\infty}  \frac{ dr }{r^{2n_3 - d+1}}  \right|  \nonumber \\
 && +  A(S^{d-1})  R_0^{2d-2} \left| \int_{L}^{+\infty}  r^{d-1}   \left( - \gamma_0 V \left(\xi_{+\infty} \right) + \sum_{n=1}^d \gamma_n C^{d}_{n} \left(1-\frac{n}{d}\right) \frac{ \sin^{2n}\xi_{+\infty}}{(rR_0)^{2n}}  \right)  dr\right| ~  .  \nonumber \\
 \label{Estaticineq} 
\end{eqnarray} 
 The first term in the right hand side of (\ref{Estaticineq}) is  finite since all $C^2$-functions, namely, $f(r)$ and $\xi(r)$ are bounded on the closed interval $I_L$.  In order to control the second and the third terms in (\ref{Estaticineq})  on the open interval $I_A$,   one has to set  $\xi_{+\infty}$, the scalar potential $V \left(\xi_{+\infty} \right)$, and  the order $n_3$ in  \eqref{eq:fasympcon} to be  $\xi_{\infty} = 0$, $V \left(\xi_{\infty} \right) = 0$, and $n_3 > \frac{d}{2} $ for finite $\tilde{\xi}_1$, respectively. In the case of $\Lambda = 0$, we just replace $\Lambda$ by $1/R_0^2$ in \eqref{Estaticineq}. One regains the same results. Comparing this energy estimate analysis with the result \eqref{eq:nsympconres}, we conclude
\begin{eqnarray} 
\xi_{\infty} = 0 ~ , \quad V \left(\xi_{\infty} \right) = 0 ~ , \quad n_3 \ge \frac{1}{2} (d+1) ~ . \label{eq:fasympcon0}
\end{eqnarray}

Moreover, in order to have a consistent result, since we have the Einstein's field equation, one has to check the estimate
\begin{eqnarray} 
\label{G00} 
&&   \int_{r_h}^{+\infty} e^\delta r^{d-1} \left( R^0_{~ 0} - \frac{1}{2} \delta^0_{~ 0}  R +  \delta^0_{~ 0} \Lambda \right) dr \nonumber\\
  && =\frac{1}{2}  \int_{r_h}^{+\infty} e^\delta r^{d-1}  \left( (d-1)\left( \frac{f'}{r}+\frac{(d-2)}{r^2} f \right)-\frac{(d-1)(d-2)}{ r^2 R^2_0} + 2 \Lambda  \right) dr ~ . \nonumber\\
\end{eqnarray}
Applying the expansions \eqref{eq:Expandasymp}  and \eqref{eq:fasympcon}, and repeating similar steps as in \eqref{Estaticineq}, we obtain that the integral \eqref{G00} is finite. In other words, the finite energy black holes do exist.

\indent Thus, we could state 
\begin{theorem}
\label{MainresultExistence}
Suppose ${\bf{w} }(r)$ be a solution of \eqref{ReducedEinsteineqSkyrmionEq1} with the initial value $\mathbf{w}_h$ and  $\Lambda \le 0$. Then, we have a family of global solutions with finite energy satisfying \eqref{eq:horizonexpand}, \eqref{eq:Expandasymp},  \eqref{eq:fasympcon},  \eqref{eq:nsympconres}, and \eqref{eq:fasympcon0}  that  connects  two boundaries, namely the horizon and the asymptotic regions.
\end{theorem}

\section{Linear Dynamical Stability Analysis}
\label{sec:linearstableanal}

In this section, we will discuss the linear dynamical stability analysis of the models using similar method as in \cite{Shiiki:2005aq}. First, we take in advance the metric functions  in the preceding sections to be time dependent, namely,  $\delta_t \equiv \delta(r, t)$, $m_t \equiv m(r, t)$, and $\xi_t \equiv \xi(r, t)$. Then, we expand them around $\delta(r)$, $m(r)$, and $\xi(r)$ where $\delta(r)$, $m(r)$, and $\xi(r)$ are the static solutions of \eqref{eq:eomSkyrme}, \eqref{eq:Einsteineq2}, and \eqref{eq:constraints}.

Let us write down the energy functional of the models in which we have $\delta(r, t)$, $m(r, t)$, and $\xi(r, t)$ 
\begin{equation}\label{eq:EeffSkyrmetdep}
    E(t) =\int d^{d}x \sqrt{-g^{(d+1)}} \left( u_t \left(\frac{(\dot{\xi}_t)^2}{e^{2\delta_t}f_t} + f_t  (\xi'_t)^2 \right) +v_t \right) ~ ,
\end{equation}
where $u_t$ and $v_t$ have the same form as defined in \eqref{eq:uv} but we replace $\xi$ by $\xi_t$. The variation of  \eqref{eq:EeffSkyrmetdep} with respect to $\xi_t$ gives  the Skyrmion equation of motions 
\begin{eqnarray}
 \left(\xi''_t -\frac{\ddot{\xi}_t}{e^{2\delta_t}f_t^2}\right)+\left(\frac{\delta'_t f_t + f'_t}{f_t}+\frac{d-1}{r}+\frac{1}{u_t}\frac{\partial u_t}{\partial r}\right)\xi'_t +\left(\frac{\dot{\delta}_t f_t +\dot{f}_t}{e^{2\delta_t}f_t^3}\right)\dot{\xi}_t \nonumber\\
 +\left(\frac{1}{2u_t}\frac{\partial u_t}{\partial \xi_t}\right)\left( (\xi'_t)^2-\frac{(\dot{\xi}_t)^2}{e^{2\delta_t}f_t^2}\right) = \frac{1}{2u_t f_t}\frac{\partial v_t}{\partial\xi_t} ~. \label{eq:eomSkyrmet}
\end{eqnarray}

In this case, the components of the Einstein field equations in \eqref{eq:Einsteineq}-\eqref{eq:Einsteineq2} are  modified, that is,
\begin{equation}\label{eq:Einsteineqt}
    (d-1)\frac{f_t}{2}\left( \frac{f'_t}{rf_t}+\frac{d-2}{r^2}\right) -\frac{(d-1)(d-2)}{2(rR_0)^2}+\Lambda  = - u_t  \left(\frac{(\dot{\xi}_t)^2}{e^{2\delta_t}f_t}+ f_t (\xi'_t)^2\right) - v_t    ~ ,   
\end{equation}
\begin{equation}\label{eq:Einsteineq1t}
 (d-1)\frac{f_t}{2}\left[\frac{2\delta'_t f_t +f'_t}{rf_t}+\frac{d-2}{r^2}\right]-\frac{(d-1)(d-2)}{2(rR_0)^2}+\Lambda  =  u_t  \left(\frac{(\dot{\xi}_t)^2}{e^{2\delta_t}f_t}+ f_t (\xi'_t)^2\right) - v_t ~ ,
\end{equation}
\begin{eqnarray}
    &&\frac{e^{-2\delta_t}}{2f_t^3}\left( f_t \ddot{f}_t - 2(\dot{f}_t)^2- f_t \dot{f}_t \dot{\delta}_t \right) +\frac{f_t}{2} \left( 2\delta''_t +2(\delta'_t)^2+\frac{f''_t +3\delta'_t f'_t}{f_t}\right) \nonumber\\
    && +\frac{d-2}{r}\left(f'_t+\delta'_t f_t \right)-\frac{(d-2)(d-3)}{2(rR_0)^2}\left(1-f_t R_0^2\right)+\Lambda\nonumber\\
     &&  \quad  = ~  \frac{1}{d (d-1)} \sum_{n=1}^d\gamma_n C^{d}_{n} ~ n \left(2n-d-1\right) \frac{\sin^{2(n-1)}\xi_t}{(rR_0)^{2(n-1)}}   \left(-\frac{(\dot{\xi}_t)^2}{e^{2\delta_t}f_t}+ f_t (\xi'_t)^2\right)  \nonumber\\
 && \quad \quad - \frac{1}{d (d-1)} \sum_{n=1}^d\gamma_n C^{d}_{n} \left( 2n^2-3nd+n-d+d^2\right) \frac{\sin^{2n}\xi_t}{(rR_0)^{2n}}  - \gamma_0 V  ~ . \label{eq:Einsteineq2t}
\end{eqnarray}
Equations  \eqref{eq:Einsteineqt} and  \eqref{eq:Einsteineq1t} can further be simplified  into
\begin{eqnarray}
    m'_t &=& \frac{(d-2)}{(d-1)}r^{d-1}\left( u_t \left(\frac{(\dot{\xi}_t)^2}{e^{2\delta_t}f_t}+f_t (\xi'_t)^2\right)+ v_t \right) ~ , \nonumber\\
    \delta'_t &=& \frac{2u_t r}{(d-1)}\left(\frac{(\dot{\xi}_t)^2}{e^{2\delta_t}f_t^2}+(\xi'_t)^2\right)  ~ .\label{eq:constraintst}
\end{eqnarray}
Then,  there are small fluctuation  of $\delta_t(r, t)$, $m_t(r, t)$, and $\xi_t(r, t)$ around the static classical solutions such that we have the expansion 
\begin{eqnarray}
m_t(r,t) &=& m(r)+\varepsilon ~ m_l(r,t)  ~ , \nonumber\\ 
   \delta_t(r,t) &=& \delta(r)+\varepsilon ~ \delta_l(r,t) ~ , \nonumber\\ 
   \xi_t(r,t) &=& \xi(r)+\varepsilon ~ \xi_l(r,t) ~ ,  \label{eq:expandstatsol}
\end{eqnarray}
where  \(\varepsilon > 0\) is a small  parameter. First, substituting the first and the second equations in \eqref{eq:expandstatsol} to \eqref{eq:constraintst} and taking the first order approximation on \(\varepsilon\), it yields
\begin{eqnarray}
m_l &=& \frac{d-2}{d-1}2r^{d-1}f e^{\delta}\xi'  u ~ \xi_l ~ , \nonumber\\
 \delta'_l  &=& \frac{2r}{d-1}\left( (\xi')^2  \frac{\partial u}{\partial \xi}  \xi_l  + 2  u \xi' ~ \xi'_l \right) ~ , \label{eq:expandstatsolsol}
\end{eqnarray}
where $u \equiv u_t(\xi)$, $ \frac{\partial u}{\partial \xi}  \equiv  \frac{\partial u_t}{\partial \xi}(\xi)$, and we have used \eqref{eq:eomSkyrme} in the computation to get the first equation in \eqref{eq:expandstatsolsol}.

Again,  substituting the last equation in \eqref{eq:expandstatsol} to \eqref{eq:eomSkyrmet}  and employing the computation up to the first order of \(\varepsilon\), we get 
\begin{equation}
    \frac{ u}{e^{\delta}f}\ddot{\xi}_l =\frac{1}{r^{d-1}}\left( u f r^{d-1}e^{\delta}\xi_l'\right)'+K(\xi, m,\delta, r)\xi_l ~ , \label{eq:eomSkyrmetorder1}
\end{equation}
with 
\begin{eqnarray}
    K(\xi,m,\delta,r) &\equiv& \frac{4}{d-1}r e^{\delta} (\xi')^3   \frac{\partial u}{\partial\xi}   u f +\frac{1}{r^{d-1}}\left(r^{d-1}e^{\delta}f \xi' \left(\frac{\partial u}{\partial\xi}   -\frac{4}{d-1} u^2  \right)\right)'\nonumber\\
    &&-\frac{e^{\delta}}{2}\left(f  \frac{\partial^2 u}{\partial\xi^2} ( \xi')^2 +  \frac{\partial^2 v}{\partial\xi^2}   \right) ~ .
\end{eqnarray}
Taking
\begin{equation}
    \xi_l (r, t) =  \left(u f e^{\delta} r^{d-1} \right)^{-1/2} \psi(r) e^{i\omega t}~ ,
\end{equation}
we can cast  \eqref{eq:eomSkyrmetorder1} into Sturm-Liouville equation
\begin{equation}
   \psi'' + \left(-\frac{1}{2}  \left(u f e^{\delta} r^{d-1} \right)''+\frac{ \left( (u f e^{\delta} r^{d-1})' \right)^2}{4 u f e^{\delta} r^{d-1} } + r^{d-1} \frac{ e^{\delta}f }{u}  K(\xi,m,\delta,r) + \omega^2\right) \psi = 0  ~ . \label{eq:eomSkyrmetorder1lg}
\end{equation}
The solution $\psi(r)$ is said to be stable if the eigenvalue $\omega^2 > 0$ and  $\psi(r) > 0$ for $r_h < r < +\infty$. However, to have unstable solutions,  it is sufficient to show that there exists an eigenvalue with $\omega^2<0$. In other words, we have stability properties of the solutions  in the context of Sturm-Liouville theory.

It is straightforward to show that there exists a unique local solution of \eqref{eq:eomSkyrmetorder1lg} in an arbitrary interval $I \subset I_\infty$ using the same procedure as in subsection \ref{subsec:LocExisSmooth}. Then, we apply the uniqueness condition to get a global solution in  $  I_\infty$ by gluing all of these local solutions.

Let us first discuss the behavior of \eqref{eq:eomSkyrmetorder1lg} in the asymptotic region where  $r \in I_A $.  First of all, we take the case of $\Lambda < 0$. Employing the expansion \eqref{eq:Expandasymp}  and \eqref{eq:fasympcon}, eq. \eqref{eq:eomSkyrmetorder1lg} can be simplified to
\begin{equation}
   \psi'' +  \frac{\Lambda \gamma_0}{d(d-1) \gamma_1} \frac{\partial^2 V}{\partial \xi^2}(\xi_\infty)    r^{d+1}   \psi = 0  ~ , \label{eq:eomSkyrmetorder1lgasymp}
\end{equation}
where we have set $n_3 = d-1$ implying $d \ge 3$. At the lowest order of the expansion, we could have either $\omega^2 \in \lR$ for $\tilde{\xi}_1  \in \lR$, or
\begin{equation}
 \omega^2  = - \frac{16 \Lambda^2 \tilde{\xi}_1 \gamma_1}{d^2 (d-1)^2}  ~ , \label{eq:omega}
\end{equation}
 for $\tilde{\xi}_1  \in \lR$. Assuming 
\begin{equation}
  \frac{\partial^2 V}{\partial \xi^2}(\xi_\infty) > 0   ~ , \label{eq:locminVasymp}
\end{equation}
 the solution of \eqref{eq:eomSkyrmetorder1lgasymp} has the form
\begin{equation}
   \psi(r, \Lambda) = A_0 \left(\frac{d+3}{k_\Lambda^{1/2} r^{\frac{d+1}{2}} } \right)^{1/2} {\mathrm{exp}}\left(-\frac{2 k_\Lambda^{1/2}}{d+3}  r^{\frac{d+3}{2}} \right) \left(1 + O\left( r^{-\frac{d+3}{2}} \right) \right)  ~ , \label{eq:soleomSkyrmetorder1lgasymp}
\end{equation}
which is the asymptotic form of the modified Bessel of the second kind, where
\begin{equation}
  k_\Lambda \equiv  \frac{|\Lambda | \gamma_0}{d(d-1) \gamma_1} \frac{\partial^2 V}{\partial \xi^2}(\xi_\infty)   ~ ,
\end{equation}
and $A_0 > 0$. It is easy to see that  $\psi(r)$ in \eqref{eq:soleomSkyrmetorder1lgasymp} is a positive definite function on $I_A$. In the $\Lambda = 0$ case, we have a similar equation as \eqref{eq:eomSkyrmetorder1lgasymp} with $\omega^2 \in \lR$ whose solution is given by
\begin{equation}
   \psi(r,0) = A_0 \left(\frac{d+1}{k_0^{1/2} r^{\frac{d-1}{2}} } \right)^{1/2} {\mathrm{exp}}\left(-\frac{2 k_0^{1/2}}{d+1}  r^{\frac{d+1}{2}} \right) \left(1 + O\left( r^{-\frac{d+1}{2}} \right) \right)  ~ , \label{eq:soleomSkyrmetorder1lgasymp1}
\end{equation}
which has positive value on $I_A$ with
\begin{equation}
  k_0 \equiv  \frac{ \gamma_0}{2 \gamma_1 R_0^2} \frac{\partial^2 V}{\partial \xi^2}(\xi_\infty)   ~ ,
\end{equation}
satisfying \eqref{eq:locminVasymp}.

Next, we show that the existence of positive definite solutions of  \eqref{eq:eomSkyrmetorder1lg} on  $I_L $. In our proof, we use  the Fixed Point Theorem:
\begin{theorem}{\textnormal{\cite{Krasnosel_1964}}}
\label{FixedPoint}
Let us consider ${\mathcal E}$ to be a Banach space in which we have a cone  ${\mathcal K} \subset {\mathcal E}$. Assuming both $\Omega_1$ and $\Omega_2$ to be open subsets of  ${\mathcal E}$ with  $0 \in \Omega_1$ and  $\bar{\Omega}_1 \subset \Omega_2$, and let $\hat{H}$ be a completely continuous operator satisfying
\begin{equation}  
\hat{H} : {\mathcal K} \cap \left( \bar{\Omega}_2 \diagdown \Omega_1 \right) \to {\mathcal K} ~ ,
\end{equation}
such that either
\begin{itemize}
\item[i)] $\lVert \hat{H}\psi \rVert \le \lVert \psi \rVert$, $\psi \in  {\mathcal K} \cap  \Omega_1 $, and  $\lVert \hat{H}\psi \rVert \ge \lVert \psi \rVert$, $\psi \in  {\mathcal K} \cap  \Omega_2 $;
\item[] or
\item[ii)] $\lVert \hat{H}\psi \rVert \ge \lVert \psi \rVert$, $\psi \in  {\mathcal K} \cap  \Omega_1 $, and  $\lVert \hat{H}\psi \rVert \le \lVert \psi \rVert$, $\psi \in  {\mathcal K} \cap  \Omega_2 $.
\end{itemize}
Then, we have a fixed point of $\hat{H}$ in ${\mathcal K} \cap \left( \bar{\Omega}_2 \diagdown \Omega_1 \right)$.
\end{theorem}
In addition, the logic of some following steps follow closely \cite{Wang_1994}. The first step is to consider  $\psi(r)$ in \eqref{eq:eomSkyrmetorder1lg} in the interval $[r_1, r_2]$ with boundary conditions
\begin{eqnarray}
a_0 \psi(r_1) - a_1 \psi'(r_1) &=& 0 ~ , \nonumber\\
 b_0 \psi(r_2) -  b_1 \psi'(r_2) &=& 0  ~ , \label{eq:boundcond}
\end{eqnarray}
where $a_0, a_1, b_0, b_1 \ge 0$. Suppose we have a solution of \eqref{eq:eomSkyrmetorder1lg}
\begin{equation}  
\psi(r) = \int_{r_1}^{r_2} G(r,s) F(s)\psi(s) ds \equiv \hat{H}\psi(r) ~ ,\quad \psi \in  C[r_1, r_2] ~ , \label{eq:soleomSkyrmetorder1lg}
\end{equation}
where 
\begin{equation}  
F(r) \equiv  -\frac{1}{2}  \left(u f e^{\delta} r^{d-1} \right)''+\frac{ \left( (u f e^{\delta} r^{d-1})' \right)^2}{4 u f e^{\delta} r^{d-1} } + r^{d-1} \frac{ e^{\delta}f }{u}  K(\xi,m,\delta,r) + \omega^2 ~ ,
\end{equation}
and $ G(r,s)$ is the Green's function of
\begin{equation}  
\psi'' = 0 ~ , \label{eq:soleomSkyrmetordertriv}
\end{equation}
with boundary conditions \eqref{eq:boundcond} whose form is given by
\begin{equation} \label{eq:Greenfunc}
\begin{aligned}[b]
G(r,s) = 
&  
\begin{cases}
\frac{1}{\rho} X(r) Y(s) & ; ~ r_1 \le s \le r \le r_2 ~ , \\
 \frac{1}{\rho} X(s) Y(r)  & ; ~  r_1 \le r \le s \le r_2 ~ ,
\end{cases}
\end{aligned} 
\end{equation}
where $\rho > 0$, 
\begin{equation}  
X(r) \equiv  b_0 \left(r_2 -r \right) + b_1  ~ , \quad Y(r) \equiv  a_0 r + a_1 ~ , \quad  [r_1, r_2] ~ .
\end{equation}
Suppose we have a cone ${\mathcal K}$ in $C[r_1, r_2]$ given by
\begin{equation}  
{\mathcal K} \equiv  \left\{ \psi \in C[r_1, r_2] : \psi(r) \ge 0, \quad \min\limits_{[r_{5/4}, r_{7/4}] } \psi(r) \ge C_G \lVert \psi \rVert  \right\} ~ ,
\end{equation}
where $r_{5/4} > r_1$, $r_{7/4} < r_2$, $\lVert \psi \rVert \equiv \sup\limits_{[r_1, r_2]} | \psi(r)| $,
\begin{equation}  
C_G \equiv  \min \left\{ \frac{(r_2-r_{7/4}) b_0 + b_1}{ r_2 (b_0 + b_1)},  \frac{r_1 a_0 + a_1}{ r_2 (a_0 + a_1)}\right\} ~ .
\end{equation}
Since $G(r,s) \le G(s,s)$ for $r_1 \le r, ~ s \le r_2$, if $\psi \in {\mathcal K}$, then we have
\begin{equation}  
 \hat{H}\psi(r) = \int_{r_1}^{r_2} G(r,s) F(s)\psi(s) ds  \le \int_{r_1}^{r_2} G(s,s) F(s)\psi(s) ds   ~ , \label{eq:ineq1}
\end{equation}
implying
\begin{equation}  
\lVert \hat{H}\psi(r) \rVert  \le \int_{r_1}^{r_2} G(s,s) F(s)\psi(s) ds   ~ .
\end{equation}
Moreover, since
\begin{equation}  
G(r,s) \ge C_G ~ G(s,s)   ~ , \quad   [r_{5/4}, r_{7/4}] ~ ,
\end{equation}
it follows
\begin{equation}  
 \min\limits_{ [r_{5/4}, r_{7/4}] } \hat{H}\psi(r) \ge C_G \lVert \hat{H}\psi \rVert  ~ .
\end{equation}
Thus, $\hat{H}{\mathcal K} \subset {\mathcal K}$ which implies that the mapping  $\hat{H}: {\mathcal K} \to {\mathcal K}$ is completely continuous.

Suppose there is a constant $C_1 > 0$ and $0 < \psi \le C_1$ such that
\begin{equation}  
 \int_{r_1}^{r_2} G(s,s) F(s)  ds  \le 1 ~ . \label{eq:ineq2}
\end{equation}
If $\psi \in {\mathcal K}$ and $ \lVert \psi \rVert = C_1$, then from \eqref{eq:ineq1} and \eqref{eq:ineq2}  it follows
\begin{equation}  
 \hat{H}\psi(r)   \le \int_{r_1}^{r_2} G(s,s) F(s)\psi(s) ds \le \lVert \psi \rVert  ~ . \label{eq:ineq3}
\end{equation}
Defining
\begin{equation}  
\Omega_1 \equiv  \left\{ \psi \in{\mathcal E} : \lVert \psi \rVert  < C_1\right\} ~ ,
\end{equation}
\eqref{eq:ineq3} implies
\begin{equation}  
  \lVert \hat{H}\psi \rVert    \le \lVert \psi \rVert  ~ , \quad \psi \in {\mathcal K}\cap \partial\Omega_1 ~ . 
\end{equation}
Next, we also have $C_2 > 0$ and $\psi \ge C_2$ such that 
\begin{equation}  
C_G  \int_{r_{5/4}}^{r_{7/4}} G(r_{3/2},s) F(s)  ds  \ge 1 ~ , \label{eq:ineq4}
\end{equation}
where $r_{5/4}< r_{3/2} < r_{7/4}$. Introducing $C_3 \equiv \max\left\{\frac{C_1}{r_{7/4} -r_{5/4}}, \frac{C_2}{C_G} \right\} $ and
\begin{equation}  
\Omega_2 \equiv  \left\{ \psi \in{\mathcal E} : \lVert \psi \rVert  < C_3\right\} ~ ,
\end{equation}
if  $\psi \in {\mathcal K}$ and $ \lVert \psi \rVert = C_3$, then
\begin{equation}  
 \min\limits_{ [r_{5/4}, r_{7/4}] } \psi(r) \ge C_G \lVert \psi \rVert \ge   \lVert \psi \rVert  ~ ,
\end{equation}
such that 
\begin{eqnarray}  
 \hat{H}\psi(r_{3/2}) = \int_{r_1}^{r_2} G(r_{3/2},s) F(s)\psi(s) ds  \ge C_G  \lVert \psi \rVert  \int_{r_{5/4}}^{r_{7/4}} G(r_{3/2},s) F(s)  ds  \ge   \lVert \psi \rVert  ~ . \label{eq:ineq5}
\end{eqnarray}
So, we obtain
\begin{equation}  
  \lVert \hat{H}\psi \rVert    \ge \lVert \psi \rVert  ~ , \quad \psi \in {\mathcal K}\cap \partial\Omega_2 ~ . 
\end{equation}

To conclude, using i) of Theorem \ref{FixedPoint}, the operator $ \hat{H}$ has a fixed point  in ${\mathcal K} \cap \left( \bar{\Omega}_2 \diagdown \Omega_1 \right)$ with $C_1 \le  \lVert \psi \rVert  \le C_3$. Moreover, the fact $G(r,s) > 0$ implies $\psi(r) > 0$  in the interval $[r_1, r_2]$.   It is worth  mentioning that by interchanging $\Omega_1 \to \Omega_2$ and $\Omega_2 \to \Omega_1$ in the above computation, and using  ii) of Theorem \ref{FixedPoint}, we will obtain the same results.

Since we have proved that $\psi(r) > 0$  in an arbitrary interval $[r_1, r_2] \subset I_L$ which is also unique and $\psi(r)$ is at least $C^ 2$-function, we could extend the proof by gluing all of these solutions to have $\psi(r) > 0$ in $I_L$. Again, by gluing  $\psi(r) > 0$ in $I_L$ with \eqref{eq:soleomSkyrmetorder1lgasymp}, we can finally obtain $\psi(r) > 0$ in $I_\infty$. 

So, we have
\begin{theorem}
\label{MainresultStability}
There exists dynamically stable and unstable static spherical symmetric solutions of higher dimensional Einstein-Skyrme system with general couplings \eqref{eq:EinsteinSkyrmeLdgenform} and $\Lambda \le 0$.
\end{theorem}
\section*{Acknowledgments}
The work in this paper is  supported by Riset ITB 2022  and PDUPT Kemendikbudristek-ITB 2022.

\bibliography{referenceER}
\bibliographystyle{unsrt}
\end{document}